\documentclass[11pt,reqno]{amsart}

\usepackage[english]{babel}
\usepackage{amssymb,amsmath,amsthm,amscd,latexsym,amsfonts}
\usepackage{mathtools}
\usepackage[T1]{fontenc}
\usepackage{graphicx,epstopdf}
\usepackage{graphicx}
\usepackage{xcolor}
\usepackage{cite}
\usepackage{MnSymbol}
\usepackage[a4paper,top=3cm, bottom=3cm, left=3.1cm, right=3.1cm]{geometry}

\newtheorem{thm}{Theorem}
\newtheorem{defn}{Definition}
\newtheorem{lemma}{Lemma}
\newtheorem{pro}{Proposition}
\newtheorem{rk}{Remark}
\newtheorem{cor}{Corollary}

\numberwithin{equation}{section} 

\newcommand{\bea}{\begin{eqnarray}}
	\newcommand{\eea}{\end{eqnarray}}

\newcommand{\Z}{\mathbb{Z}}

\def\Z{\mathbb{Z}}

\begin{document}
	\title [TIGMs for Hard Core model ]
	{Translation-Invariant Gibbs measures for the Hard Core model with a countable set of spin values}
	
	\author {R.M. Khakimov, M.T. Makhammadaliev}
	
	\
	\address{R.M. Khakimov$^{a,b,c}$, M.T. Makhammadaliev$^{a,b}$
		\begin{itemize}
			\item[$^a$] V.I.Romanovskiy Institute of Mathematics,  9, Universitet str., 100174, Tashkent, Uzbekistan;
			\item[$^b$] Namangan State  University, Namangan, Uzbekistan;
            \item[$^c$] New Uzbekistan University, 54 Mustaqillik Ave., Tashkent, 100007, Uzbekistan.
	\end{itemize}}
	\email{rustam-7102@rambler.ru, mmtmuxtor93@mail.ru}

	\begin{abstract}
	In this paper, we study the Hard Core (HC) model with a countable set $\mathbb Z$ of spin values on a Cayley tree
	 of order $k=2$. This model is defined
by	a countable set of parameters (that is, the activity function $\lambda_i>0$, $i\in \mathbb Z$).
	A functional equation is obtained that provides the consistency condition for
	finite-dimensional Gibbs distributions. Analyzing this equation,
	the following results are obtained:
	\begin{itemize}
		\item[-] Let $k\geq 2$ and $\Lambda=\sum_i\lambda_i$. For  $\Lambda=+\infty$ there is no translation-invariant Gibbs measure (TIGM);
		\item[-] Let $k=2$ and $\Lambda<+\infty$. For the model under constraint such that at $G$-admissible graph the loops are imposed at two vertices of the graph, the uniqueness of TIGM is proved;
		\item[-] Let $k=2$ and $\Lambda<+\infty$. For the model under constraint such that at $G$-admissible graph the loops are imposed at three vertices of the graph, the uniqueness and non-uniqueness conditions of TIGMs are found.
	\end{itemize}
		
	\end{abstract}
	\maketitle
	
	
	
\section{Introduction}

 The theory of Gibbs measures is well developed in many classical models from physics (for example, the Ising model, the Potts model, the HC model), when the set of spin values is a finite set.  It is known that each limiting Gibbs measure for the lattice models corresponds to a certain phase of a physical system. A phase transition problem, or a change in the physical system's state when temperature varies, is one of the most difficult topics in the theory of Gibbs measures. When the Gibbs measure is not unique, the phase transition occurs \cite{6}.

There are papers devoted to the study of (gradient) Gibbs measures for models with an infinite set of spin values. For gradient Gibbs measures of gradient potentials on Cayley tree see \cite{HKR}, \cite{HK}, \cite{Z} and references therein.
In \cite{GR} the uniqueness of the translation-invariant Gibbs measure for the antiferromagnetic Potts model with a countable set of
spin values and a nonzero external field was shown. In \cite{G} for this Potts model the Poisson measures, which are Gibbs measures, were described.

We consider a HC model on Cayley trees. The use of the Cayley tree is motivated (see \cite[page 18]{Robp} and references therein) by the applications, such as information flows and
reconstruction algorithms on networks, evolution of genetic
data and phylogenetics, DNA strands and Holliday junctions, or computational complexity on graphs.

Many papers are devoted to the study of limit Gibbs measures for Hard Core models with a finite set of spin values (see, for example, \cite{RM}, \cite{RKhM},\cite{SR} and the references therein).

 In this paper, we study HC model with a countable set of spin values. Our motivation is that there are biological and physical systems configurations (states) of which defined by spins with a countable set of values.  The main examples of such spin systems are harmonic
oscillators. Another example is the Ginzburg-Landau interface model; which is obtained from the anharmonic oscillators (see \cite{FS}, \cite{G}, \cite{Ve}). In \cite{RMR}, the HC model with a countable set of spin values was studied for the first time and conditions for the existence of Gibbs measures are found. Moreover, the exact value $\Lambda_{\rm cr}$  of the parameter $\Lambda$ is found, (where $\Lambda$ is the sum of the series obtained from the sequence of parameters $\{\lambda_j\}_{j\in \mathbb Z}$), such that for $\Lambda\leq\Lambda_{\rm cr}$ there is exactly one periodic Gibbs measure which is translation-invariant, and for $\Lambda>\Lambda_{\rm cr}$ there are exactly three periodic Gibbs measures, one of which is translation-invariant.

In this paper, we prove the existence of TIGMs for the HC model with a countable set of spin values on a Cayley tree of order two under some conditions, and also find the conditions of the uniqueness and non-uniqueness of such measures.

\section{Preliminaries}

The Cayley tree $\Im^k$ of order $ k\geq 1 $ is an infinite tree, i.e. graph without cycles, each vertex of which has exactly $k+1$ edges, where $V$ is the set of vertices $\Im^k$, $L$ is the set of edges.  If  $l \in L$ an edge with  endpoints  $x, y\in V$ then we write $l=\langle x,y\rangle $ and the endpoints are called nearest neighbors.

For fixed $x^0\in V$
$$W_n=\{x\in V\,| \, d(x,x^0)=n\}, \qquad V_n=\{x\in V\,| \, d(x,x^0)\leq n\},$$
 where $d(x,y)$ is the distance between the vertices $x$ and $y$ on the Cayley tree.

Write $x\prec y$ if the path from $x^0$ to $y$ goes through $x$. A vertex $y$ is called a direct successor of a vertex $x$ if $y\succ x$ and $x,y$ are nearest neighbors. The set of direct successors of the vertex $x$ will be denoted by $S(x)$.

We consider the Hard Core (HC) model with a countable set of spin values in which the spin variables take values in the set of  integers $\mathbb{Z}$, and are located at the tree vertices. A configuration $\sigma = \{\sigma(x) \, |\, x \in  V \}$  is then defined as a function $\sigma=\{\sigma(x)\in \mathbb Z: x\in V\}$. 

We consider the set $\mathbb{Z}$ as the set of vertices of a graph $G$.
We use the graph $G$  to define a $G$-admissible configuration as follows.
A configuration $\sigma$ is called a
$G$-\textit{admissible configuration} on the Cayley tree (in a subset $A\subset V$), if $\{\sigma (x), \, \sigma (y)\}$ is one edge of the graph $G$
for any pair of nearest neighbors $x,y$ in $V$ (in $A$). We
let $\Omega^G$ ($\Omega_A^G$) denote the set of $G$-admissible configurations $\sigma$ (resp. $\sigma_A$).

The activity set \cite{bw} for a graph $G$ is the bounded function $\lambda:G
\to R_+$ from the vertices of $G$  to the set of positive real
numbers. The value $\lambda_i$ of the function $\lambda$ at the vertex
$i \in \mathbb{Z}$ is called the vertex activity.

For given $G$ and $\lambda$ we define the Hamiltonian of the $G-$HC model as
\begin{equation}\label{H} H^{\lambda}_{G}(\sigma)=\left\{%
\begin{array}{ll}
     \sum\limits_{x\in{V}}{\ln\lambda_{\sigma(x)},} \  \mbox{if} \ \sigma \in\Omega^G , \\
   +\infty ,\  \  \  \ \ \ \ \ \ \   \ \mbox{if} \ \sigma \ \notin \Omega^G . \\
\end{array}
\right.
\end{equation}

For nearest-neighboring interaction  potential $\Phi=(\Phi_b)_b$, where
$b=\langle x,y \rangle$ is an edge,  define symmetric transfer matrices $Q_b$ by
\begin{equation}\label{Qd}
	Q_b(\omega_b) = e^{- \big(\Phi_b(\omega_b) + | \partial x|^{-1} \Phi_{\{x\}}(\omega(x)) + |\partial y |^{-1} \Phi_{\{y\}} (\omega(y)) \big)},
\end{equation}
where $\partial x$ is the set of all nearest-neighbors of $x$ and $|S|$ denotes the number of elements of the set $S$. Note that for the Cayley tree of order $k\geq 1$ we have $| \partial x|=| \partial y|=k+1$.

To introduce the notion of \emph{translations} on the Cayley tree
$\Gamma^k$, one uses its group representation $G_k$ which is  free group with
generators $a_1, \dots,\allowbreak a_{k+1}$ of order $2$ each (i.e.,
$a_i^{-1}=a_i$). It is known (see, for example, \cite[Section~2.2]{R})
that the vertices of the Cayley tree is
in a one-to-one correspondence with the elements of the group $G_k$.
Consider the family of \emph{left} shifts $T_g\colon
G_k\to G_k$ ($g\in G_k$) defined by $T_g(a):=g a$, $ a\in G_k$.

This group is used to define translation-invariance of functions defined on the vertices
of the Cayley tree. In particular, the following definition is used
\begin{defn}\label{tid}  The potential $\Phi=(\Phi_b)_b$ is called invariant under a subgroup $\hat G_k\subset G_k$ of
	translations if for any $g\in \hat G_k$ and $b=\langle x, y\rangle$ one has
	$\Phi_{gb}(g \omega_{b})=\Phi_b(\omega_b)$, where $gb=\langle gx, gy\rangle$ and $g\omega$  is defined by $g\omega(x) =
	\omega(g^{-1}x)$, $x\in G_k$.
	
	In the case $\hat G_k= G_k$ the potential is called translation-invariant.
\end{defn}
Similarly one can define translation-invariant Gibbs measures (see \cite[Section 2.3.1]{BR}).

Define the Markov (Gibbsian) specification as
$$
\gamma_\Lambda^\Phi(\sigma_\Lambda = \omega_\Lambda | \omega) = (Z_\Lambda^\Phi)(\omega)^{-1} \prod_{b \cap \Lambda \neq \emptyset} Q_b(\omega_b).
$$

Let $L(G)$ be the set of edges of a graph $G$. We let $A\equiv A^G=\big(a_{ij}\big)_{i,j=0,1,2}$ denote the adjacency
matrix of the graph $G$, i.e.,

$$a_{ij}=a_{ij}^G=%
\begin{cases} 1 \ \ \mbox{if} \ \ \{i,j\}\in L(G), \\
0 \ \ \mbox{if} \ \ \{i,j\}\notin L(G).
\end{cases}
$$

\begin{defn} (See \cite[Chapter 12]{6}, \cite{HK})
\begin{itemize}	
\item[1)] 	A family of vectors $l= \{l_{xy}\}_{\langle x, y \rangle \in L}$ with
$l_{xy}=\{l_{xy}(i):i \in \mathbb Z \} \in (0, \infty)^\mathbb Z$
is called the boundary law for the Hamiltonian (\ref{H}) if for each $\langle x, y \rangle \in L $ there exists  a constant $c_{xy}> 0 $ such that the consistency equation
\begin{equation}\label{eq:bl}
l_{xy}(i) = c_{xy}\lambda_i\prod_{z \in \partial x \setminus \{y \}} \sum_{j \in \Z} a_{ij} l_{zx}(j)
\end{equation}
holds for every $ i \in \Z $, where $\partial x-$ the set of nearest neighbors of a vertex $x$.

\item[2)]  A boundary law $l$ is said to be {\em normalisable} if and only if
\begin{equation}\label{Norm}
\sum_{i \in \Z} \Big(\lambda_i \prod_{z \in \partial x} \sum_{j \in \Z} a_{ij} l_{zx}(j) \Big) < \infty
\end{equation} at any $x \in V$.

\item[3)] A boundary law 	is called {\em $q$-height-periodic} (or $q$-periodic) if $l_{xy} (i + q) = l_{xy}(i)$
	for every oriented edge $\langle x,y \rangle $ and each $i \in \Z$.
	
\item[4)] 	 A boundary law is called {\em translation-invariant}  if it does not depend on edges of the tree, i.e., $l_{xy} (i) = l(i)$
for every oriented edge $\langle x,y \rangle $ and each $i \in \Z$.
\end{itemize}
\end{defn}
Assume $l_{xy}(0)=1$, for each $\langle x, y \rangle \in L $ (normalization at $0$), then dividing (\ref{eq:bl}) to the equality obtained for $i=0$ we get
\begin{equation}\label{ma}
l_{xy}(i) = {\lambda_i\over \lambda_0}\prod_{z \in \partial x \setminus \{y \}}{a_{i0}+ \sum_{j \in \Z_0} a_{ij} l_{zx}(j)\over a_{00}+\sum_{j \in \Z_0} a_{0j} l_{zx}(j)}.
\end{equation}

\begin{rk} \label{r2} We note that
	\begin{itemize}
\item[a.] There is an one-to-one correspondence between boundary laws
and Gibbs measures (i.e., tree-indexed Markov chains) if the boundary laws are  normalisable \cite{Z1} (see \cite[Theorem 3.5.]{HKR}).

\item[b.] In \cite{HKR} it is shown that a translation-invariant boundary law $z\in \mathbb R_+^\infty$ satisfies the condition of normalisability, if $z\in l^{\frac{k+1}{k}}$.
\end{itemize}
\end{rk}

In this paper we consider the nearest-neighboring interaction  potential $\Phi=(\Phi_b)_b$, which corresponds to the HC model  (\ref{H}), i.e., for $G$-admissible configuration $\sigma$ and $A\subset V$:
$$ \Phi_{A}(\sigma_A)=
\begin{cases}
	\lambda_{\sigma(x)}, \ \ \mbox{if} \ \ A=\{x\}, \\
	0, \ \ \ \ \ \ \  \mbox{if} \ \ |A|\geq 2,
\end{cases}	
$$
 and will study Gibbs measures of this model. By Remark \ref{r2} each normalisable boundary law $l$ defines a Gibbs measure. In this paper our aim is to find $1$-height-periodic boundary laws for the HC model for a specially chosen graph $G$ (see below). We show that these boundary laws will be normalisable and therefore define Gibbs measures.

We consider the graph $G$ with $a_{i0}=1$ for any $i\in \mathbb{Z}$, $a_{11}=1$  and $a_{im}=0$ otherwise (see Fig.\ref{fig1}).
\begin{figure}[h]
\begin{center}
   \includegraphics[width=13cm]{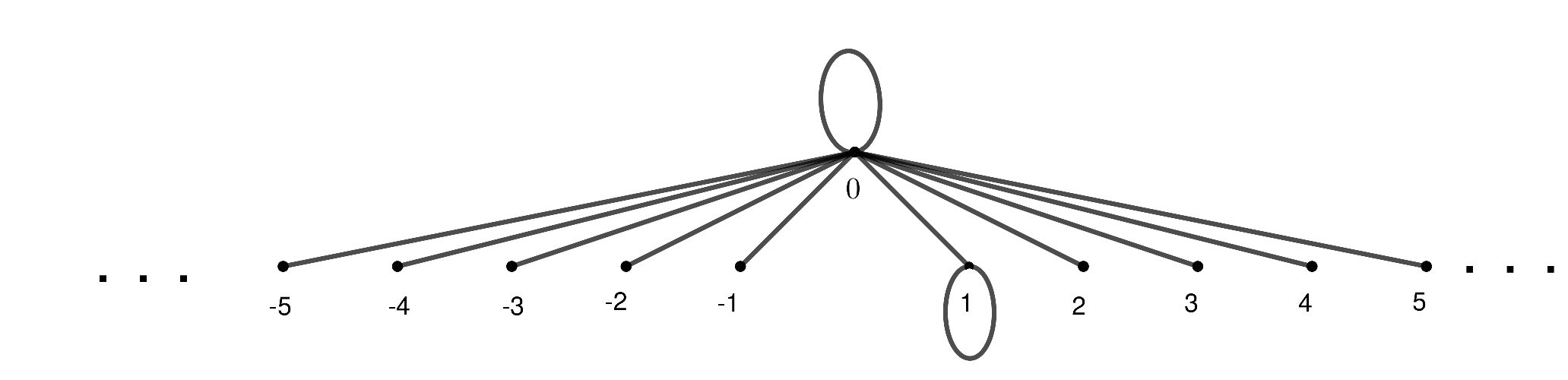}
\end{center}
     \caption{Countable graph $G$ with vertex set $\mathbb Z$ such that the loops are imposed at two vertices of the graph.}\label{fig1}
\end{figure}

For $x\in S(y)$, $y\in V$, introduce new variables as
$z_{i,x}=l_{xy}(i)$, then in case of $G$ (given in Fig.\ref{fig1}), from (\ref{ma}) (see \cite{6} and \cite{BR}) we obtain

\begin{equation}\label{e8}
\begin{cases}
z_{1,x}=\lambda_1\prod_{y \in S(x)} {1+z_{1,y}\over 1+\sum_{j\in \mathbb{Z}_0} z_{j,y}}, \\
z_{i,x}=\lambda_i\prod_{y \in S(x)} {1\over 1+\sum_{j\in \mathbb{Z}_0} z_{j,y}}, \ \ i\in \mathbb{Z}_0\setminus\{1\}.
\end{cases}
\end{equation}

\begin{rk}\label{rk1} Note that if we change the loop at vertex 1 to another arbitrary vertex (except vertex 0), we obtain a system of equations similar to \eqref{e8}, but in this case the notations may be slightly different. Therefore, it is sufficient to consider the case where there is a loop at vertex 1.
\end{rk}

\section{Translation-invariant measures for the model with two loops}

The problem of the finding of the general form of solutions of the equation (\ref{e8}) seems to be very difficult.
In this subsection, we consider translation-invariant solutions, i.e., $z_x=z=(z_i)_{i\in \mathbb Z_0}$, with $z_i\in \mathbb R_{+}.$
In this case the equation (\ref{e8}) has the following form
\begin{equation}\label{e9}
\begin{cases}
z_{1}=\lambda_1\left({1+z_1\over 1+\sum_{j\in \mathbb{Z}_0} z_{j}}\right)^k, \\
z_i=\lambda_i\left({1\over 1+\sum_{j\in \mathbb Z_0} z_j}\right)^k, \ \ i\in \mathbb{Z}_0\setminus\{1\}.
\end{cases} \end{equation}
Here $\lambda_i>0, \ z_i>0$.

\begin{lemma}\label{L1} Let $k\geq2$. If there is a positive solution $\{z_j\}_{j\in \mathbb Z_0}$ of the system of equations (\ref{e9}) for some sequence of parameters $\{\lambda_j\}_{j\in \mathbb Z_0}$ then series $\sum_{j\in \mathbb Z_0} z_j$ and $\sum_{j\in \mathbb Z_0}\lambda_j$ obtained  respectively from $\{z_i\}_{i\in \mathbb Z_0}$ and $\{\lambda_j\}_{j\in \mathbb Z_0}$ converge.
\end{lemma}
\begin{proof} Let  $\{z_j\}_{j\in \mathbb Z_0}$ be a solution of the system of equations (\ref{e9}). We assume that the series $\sum_{j\in \mathbb Z_0} z_j$ diverges. Then, since $z_j>0$, it is obvious that $\sum_{j\in \mathbb Z_0} z_j=+\infty$. Hence due to (\ref{e9}) we get $z_i=0, \ i\in \mathbb Z_0$, i.e., $\sum_{i\in \mathbb Z_0} z_i<+\infty.$ This is a contradiction. Therefore, under the conditions of lemma the series $\sum_{j\in \mathbb Z_0} z_j$  converges.

Let  $\sum_{j\in \mathbb Z_0} z_j=A$. Then from (\ref{e9}) we obtain $ A\left({1+A}\right)^k=\sum_{i\in \mathbb Z_0}\lambda_i+\lambda_1\left(\sum_{l=1}^kC_k^lz_1^l\right)$. Thus, $\sum_{i\in \mathbb Z_0}\lambda_i=A(1+A)^k-\lambda_1\left(\sum_{l=1}^kC_k^lz_1^l\right)$, i.e., the series $\sum_{i\in \mathbb Z_0}\lambda_i$ converges.
Lemma is proved.
\end{proof}

 By Lemma \ref{L1} it follows that there is no positive solution of the system of equations (\ref{e9}) for which the series $\sum_{j\in \mathbb Z_0} z_j$ and $\sum_{j\in \mathbb Z_0}\lambda_j$ diverge, i.e., these conditions are necessary for the existence of a solution (\ref{e9}).

\begin{pro}\label{P1} Let $k=2$. If the series $\sum_{j\in \mathbb Z_0}\lambda_j$ obtained from a sequence of parameters $\{\lambda_j\}_{j\in \mathbb Z_0}$ converges then for the sequence $\{\lambda_j\}_{j\in \mathbb Z_0}$ there exists a unique positive solution $\{z_j\}_{j\in \mathbb Z_0}$ of the system of equations (\ref{e9}).
\end{pro}

\begin{proof} Let the series $\sum_{j\in \mathbb Z_0}\lambda_j$ converge and its sum be $\sum_{j\in \mathbb Z_0}\lambda_j=\Lambda$. We will prove that for the sequence   $\{\lambda_j\}_{j\in \mathbb Z_0}$ there is a unique solution of the system of equations (\ref{e11}). By Lemma \ref{L1} it follows that for the existence of a solution $\{z_j\}_{j\in \mathbb Z_0}$ of the system of equations (\ref{e11}) the convergence of the series $\sum_{j\in \mathbb Z_0} z_j$ is necessary.

Let $\sum_{j\in \mathbb Z_0} z_j=A$. Then due to (\ref{e9}) we get
$$\sum_{j\in \mathbb Z_0}z_j=\frac{\sum_{j\in \mathbb Z_0}\lambda_j+\lambda_t\left(2z_t+z_t^2\right)}{(1+A)^2}.$$
Hence
\begin{equation}\label{e10}
A(1+A)^2=\Lambda+\lambda_1\left(2z_1+z_1^2\right), \ A>0.
\end{equation}

The first part of the system of equations (\ref{e9}) can be rewritten:
\begin{equation}\label{e11}
z_1(1+A)^2=\lambda_1\left(1+2z_1+z_1^2\right),
\end{equation}

Solving (\ref{e11}) with respect to $z_1$, we have:
$$z_1^{(1,2)}=\frac{(1+A)^2-2\lambda_1\pm(1+A)\sqrt{(1+A)^2-4\lambda_1}}{2\lambda_1}$$

It is easy to see that if $(1+A)^2>4\lambda_1$ then $D>0$, $z_1^{(1,2)}>0$ and $z_1^{(1)}= 1/z_1^{(2)}$.

 On substituting $z_1=z_1^{(1)}$ in (\ref{e10}), we get
$$
A(1+A)^2-\Lambda=\frac{(1+A)^2-2\lambda_1+(1+A)\sqrt{(1+A)^2-4\lambda_1}}{2}\cdot\frac{(1+A)^2+2\lambda_1+(1+A)\sqrt{(1+A)^2-4\lambda_1}}{2\lambda_1}
$$
$$\Rightarrow \ \ A(1+A)^2-\Lambda=\frac{\left((1+A)^2+(1+A)\sqrt{(1+A)^2-4\lambda_1}\right)^2-4\lambda_t^2}{4\lambda_1}
$$
\begin{equation}\label{a1}
\Rightarrow \ \ (1+A)^4+(1+A)^3\Big(\sqrt{(1+A)^2-4\lambda_1}-2\lambda_1\Big)+2\lambda_1(\Lambda-\lambda_1)=0
\end{equation}
Define
$$f(\lambda)=f(\lambda,x):=x^4+x^3\Big(\sqrt{x^2-4\lambda}-2\lambda\Big)+2\lambda(\Lambda-\lambda),  \ \ \ \lambda\in\left(0,\frac{x^2}{4}\right].$$
In this case, the equation \eqref{a1} can be written
\begin{equation}\label{a2}
f(\lambda,1+A)=0.
\end{equation}

We calculate the first and second derivatives of the function $f(\lambda)$ with respect to $\lambda$:
$$f'(\lambda)=\frac{-2x^3}{\sqrt{x^2-4\lambda}}+2\Lambda-2x^3-4\lambda \quad  \ \ \ f''(\lambda)=-\frac{4x^3}{\sqrt{(x^2-4\lambda)^3}}-4<0.$$

It is easy to see that $f(\lambda)$ is a convex function, the derivative of $f'(\lambda)$ exists in the interval $\lambda\in\left(0,\frac{x^2}{4}\right)$ and
$f( 0)=2x^4$. It follows that the function $f(\lambda)$ takes values smaller than $2x^4$ at most once. It means that the number of the solution of the equation $f(\lambda)=0$ is at most one.

On substituting $z_1=z_1^{(2)}$ in (\ref{e10}), we obtain
$$ A(1+A)^2-\Lambda=\frac{(1+A)^2-2\lambda_1-(1+A)\sqrt{(1+A)^2-4\lambda_1}}{2}\cdot\frac{(1+A)^2+2\lambda_1-(1+A)\sqrt{(1+A)^2-4\lambda_1}}{2\lambda_1}
$$
$$\Rightarrow \ \ A(1+A)^2-\Lambda=\frac{\left((1+A)^2-(1+A)\sqrt{(1+A)^2-4\lambda_1}\right)^2-4\lambda_1^2}{4\lambda_1}
$$
\begin{equation}\label{a3}\Rightarrow \ \ (1+A)^4-(1+A)^3\Big(\sqrt{(1+A)^2-4\lambda_1}+2\lambda_1\Big)+2\lambda_1(\Lambda-\lambda_1)=0.
\end{equation}
We introduce a new function
$$g(\lambda)=g(\lambda,x):=x^4-x^3\Big(\sqrt{x^2-4\lambda}+2\lambda\Big)+2\lambda(\Lambda-\lambda), \ \ \ \lambda\in\left(0,\frac{x^2}{4}\right].$$
In this case, the equation \eqref{a3} can be written
\begin{equation}\label{a4}
g(\lambda,1+A)=0.
\end{equation}

We calculate the first and second derivatives of the function $g(\lambda)$ with respect to $\lambda$:
$$g'(\lambda)=\frac{-2x^3}{\sqrt{x^2-4\lambda}}+2\Lambda-2x^3-4\lambda\quad  \ \ \ g''(\lambda)=\frac{4x^3}{\sqrt{(x^2-4\lambda)^3}}-4>0.$$

It is obvious that $g(\lambda)$ is a concave function, derivative $g'(\lambda)$ exists in the interval $\lambda\in\left(0,\frac{x^2}{4}\right)$ and $g(0)=0$. It follows that the number of the positive solutions of the equation $g(\lambda)=0$ is at most one.

On the other hand, the function
$$\phi(\lambda):=f(\lambda)-g(\lambda)=2x^3\sqrt{x^2-4\lambda}$$
is decreasing with respect to $\lambda$ and $\phi(\frac{x^2}4)=0$, i.e., $f(\frac{x^2}{4})=g(\frac{x^2}{4})$. As a result, the graph of the functions $f(\lambda)$ and $g(\lambda)$ together with the interval $[0;2x^4]$ of the $OY$ axis forms a closed line (see Fig.\ref{loop1}).
\begin{figure}[h]
\begin{center}
   \includegraphics[width=5cm]{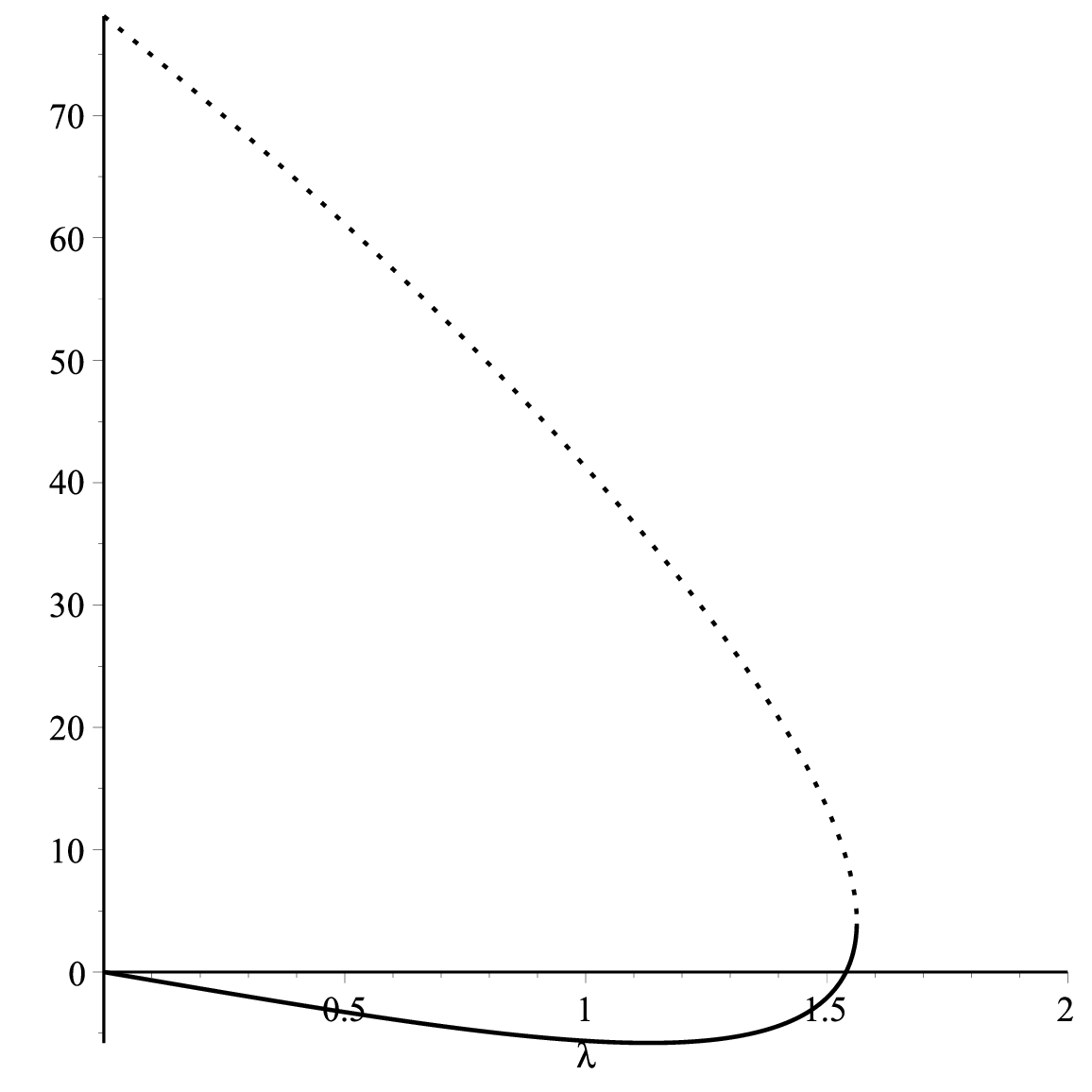}
\end{center}
     \caption{Graph functions $f(\lambda)$ (dotted line) and $g(\lambda)$ (solid line) for $x=2.5$ and $\Lambda=6$.} \label{loop1}
\end{figure}
  Then only one of the equations $f(\lambda)=0$ and $g(\lambda)=0$ would have a positive solution.
 If $f(\frac{x^2}{4})>0$ then the equation $g(\lambda)=0$ has a positive solution, if $f(\frac{x^2}{4})< 0$ then the equation $f(\lambda)=0$ has a positive solution. Hence, there is a unique $x_0$ corresponding to each $\lambda$. Accordingly, there is a unique $A_0$ corresponding to each $\lambda$.

 Therefore, the unique positive solution of the system of equations (\ref{e9}) is determined as follows:

If $f(\lambda)=0$ has a solution, then we take $z_1=z_1^{(1)}$ as a solution, if $g(\lambda)=0$ has a solution, then we take $z_1=z_1^{(2 )}$ as a solution.
The remaining coordinates are found as follows:
$$z_i={\lambda_i \over\left( 1+A_0\right)^2}, \ \ i\in \mathbb{Z}_0\setminus\{1\}.$$
\end{proof}
\begin{rk}  We note that the solution  $\{z_j\}_{j\in \mathbb Z_0}$ in Proposition \ref{P1} is normalisable because the convergence of series $\sum_{j\in \mathbb Z_0} z_j^{\frac{k+1}{k}}$ follows from the convergence of $\sum_{j\in \mathbb Z_0} z_j$. Then by Remark \ref{r2}  the Gibbs measure (denoted by $\mu_0$) corresponding to this solution exists.
\end{rk}

\section{Markov chain corresponding to the obtained TIGM in Section 3}

Below for the Gibbs measure $\mu_0$, we define a matrix $\mathbb P$ of transition probabilities. We will check the existence of a stationary distribution of the Markov chain corresponding to the measure $\mu_0$.

Consider the matrix $\mathbb P$ of transition probabilities  corresponding to the measure $\mu_0$:
$$\mathbb P =\begin{pmatrix}
\cdots & \vdots & \vdots & \vdots & \vdots & \vdots & \cdots \\
 \cdots & p_{-2-2} &  p_{-2-1}  & p_{-2 0} &  p_{-2 1} &  p_{-2 2} & \cdots \\
 \cdots & p_{-1-2} &  p_{-1-1}  & p_{-1 0} &  p_{-1 1} &  p_{-1 2} & \cdots \\
 \cdots & p_{0-2} &  p_{0-1}  & p_{0 0} &  p_{0 1} &  p_{0 2} & \cdots \\
 \cdots & p_{1-2} &  p_{1-1}  & p_{1 0} &  p_{1 1} &  p_{1 2} & \cdots \\
 \cdots & p_{2-2} &  p_{2-1}  & p_{2 0} &  p_{2 1} &  p_{2 2} & \cdots \\
\cdots & \vdots & \vdots & \vdots & \vdots & \vdots & \cdots
\end{pmatrix}.$$

Here (depending on solution $z=(z_i)_{i\in \mathbb Z_0}$)
$$p_{\sigma(x)\sigma(y)}=\frac{a_{\sigma(x)\sigma(y)}\lambda_{\sigma{(y)}} z_{\sigma(y)}}{\sum_{\sigma(y) \in \Omega^G}{a_{\sigma(x)\sigma(y)}\lambda_{\sigma{(y)}} z_{\sigma(y)}}} \ \  \Rightarrow \  \  p_{ij}=\frac{a_{ij}\lambda_{j}z_j}{\sum_{l \in \mathbb Z}{a_{il}\lambda_{l}z_l}}.$$

For the considered model, $a_{11}=1$, $a_{i0}=1$, $a_{0j}=1$ and $a_{ij}=0$ for $i\in \mathbb Z_0$ and $j \in \mathbb Z_0$. Hence
$$p_{00}=\frac{\lambda_0z_0}{\lambda_0z_0+\sum_{l \in \mathbb Z_0}{\lambda_lz_l}}=\frac{1}{1+\sum_{l \in \mathbb Z_0}{\lambda'_lz'_l}}, \ \ p_{11}=\frac{\lambda_1z_1}{\lambda_0z_0+\sum_{l \in \mathbb Z_0}{a_{1l}\lambda_lz_l}}=\frac{\lambda'_1z'_1}{1+\lambda'_1z'_1}$$
$$ p_{01}=\frac{\lambda_1z_1}{\lambda_0z_0+\sum_{l \in \mathbb Z_0}{\lambda_lz_l}}=\frac{\lambda'_1 z'_1}{1+\sum_{l \in \mathbb Z_0}{\lambda'_lz'_l}}, \ \ p_{0i}=\frac{\lambda_i z'_i}{1+\sum_{l \in \mathbb Z_0}{\lambda_lz'_l}}, $$
$$p_{10}=\frac{\lambda_0z_0}{\lambda_0z_0+\sum_{l \in \mathbb Z_0}{a_{1l}\lambda_lz_l}}=\frac{1}{1+\lambda'_1z'_1}, \ \
p_{i0}=\frac{\lambda_0z_0}{\lambda_0z_0+\sum_{l \in \mathbb Z_0}{a_{il}\lambda_lz_l}}=1.$$

 Therefore $\mathbb P$ has the following form (for $z_0=1$):
$$\mathbb P=\begin{pmatrix}
\cdots & \vdots & \vdots & \vdots & \vdots & \vdots & \cdots \\
 \cdots & 0 & 0 & 1  &  0 &  0 & \cdots \\
 \cdots & 0 & 0 & 1  & 0 &  0 & \cdots \\
 \cdots & \frac{\lambda_{-2} z_{-2}}{1+\sum_{l \in \mathbb Z_0}{\lambda_lz_l}} & \frac{\lambda_{-1} z_{-1}}{1+\sum_{l \in \mathbb Z_0}{\lambda_lz_l}} &  \frac{1}{1+\sum_{l \in \mathbb Z_0}{\lambda_lz_l}}  &  \frac{\lambda_1 z_1}{1+\sum_{l \in \mathbb Z_0}{\lambda_lz_l}}&  \frac{\lambda_2 z_2}{1+\sum_{l \in \mathbb Z_0}{\lambda_lz_l}} & \cdots \\
 \cdots & 0 & 0 & \frac{1}{1+\lambda_1z_1}  &   \frac{\lambda_1z_1}{1+\lambda_1z_1} &  0 & \cdots \\
 \cdots & 0 & 0 & 1  &  0 &  0 & \cdots \\
 \cdots & \vdots & \vdots & \vdots & \vdots & \vdots & \cdots
 \end{pmatrix}$$

We consider the vector $X=(\cdots,x_{-2},x_{-1},x_0,x_1,x_2,\cdots), \ \sum_{j\in \mathbb Z}{x_j}=1.$ If the system of equations $X\cdot\mathbb P=X$ has a solution then there exists a stationary distribution of the Markov chain corresponding to the measure $\mu_0$. So we solve the equation $X\cdot\mathbb P=X$. We have

$$X\cdot \mathbb P=\left(\cdots,\frac{x_0\lambda_{-2}z_{-2}}{1+\sum_{l \in \mathbb Z_0}{\lambda_lz_l}},\frac{x_0\lambda_{-1}z_{-1}}{1+\sum_{l \in \mathbb Z_0}{\lambda_lz_l}},p_{00},\frac{x_0\lambda_1z_{1}}{1+\sum_{l \in \mathbb Z_0}{\lambda_lz_l}}+\frac{x_1\lambda_1z_1}{1+\lambda_1z_1}, \frac{x_0\lambda_{2}z_{2}}{1+\sum_{l \in \mathbb Z_0}{\lambda_lz_l}},\cdots\right),$$
where
$$p_{00}=1-\frac{x_0\sum_{l \in \mathbb Z_0}{\lambda_lz_l}}{1+\sum_{l \in \mathbb Z_0}{\lambda_lz_l}}-\frac{x_1\lambda_1z_1}{1+\lambda_1z_1}.$$

From the equality $X\cdot\mathbb P=X$ we get
\begin{equation}\label{e13}
x_0=1-\frac{x_0\sum_{l \in \mathbb Z_0}{\lambda_lz_l}}{1+\sum_{l \in \mathbb Z_0}{\lambda_lz_l}}-\frac{x_1\lambda_1z_1}{1+\lambda_1z_1}, \  x_1=\ \frac{x_0\lambda_1z_{1}}{1+\sum_{l \in \mathbb Z_0}{\lambda_lz_l}}+\frac{x_1\lambda_1z_1}{1+\lambda_1z_1}, \ \ x_j=\frac{x_0\lambda_{j}z_{j}}{1+\sum_{l \in \mathbb Z_0}{\lambda_lz_l}},
\end{equation}
where $j\in \mathbb Z_0\setminus\{1\}.$

From the equality of (\ref{e13}) we  find $x_0$ and $x_1$:
$$x_0=\frac{1+\sum_{l \in \mathbb Z_0}\lambda_{l}z_{l}}{1+\lambda_1^2z_1^2+2\sum_{l \in \mathbb Z_0}\lambda_{l}z_{l}}, \ \
x_1=\frac{\lambda_1^2z_1^2+\lambda_1z_1}{1+\lambda_1^2z_1^2+2\sum_{l \in \mathbb Z_0}\lambda_{l}z_{l}}.$$

Using the expression for $x_0$ from the third equality of (\ref{e13}) we find $x_j$, $ j\in \mathbb Z_0\setminus\{1\}$:
$$x_j=\frac{\lambda_{j}z_{j}}{1+\lambda_1^2z_1^2+2\sum_{l \in \mathbb Z_0}\lambda_{l}z_{l}}.$$

It is easy to see that the obtained vector is stochastic:
$$\sum_{j\in \mathbb Z}{x_j}=\frac{1+\sum_{l \in \mathbb Z_0}\lambda_{l}z_{l}}{1+\lambda_1^2z_1^2+2\sum_{l \in \mathbb Z_0}\lambda_{l}z_{l}}+\frac{\lambda_1^2z_1^2+\lambda_1z_1}{1+\lambda_1^2z_1^2+2\sum_{l \in \mathbb Z_0}\lambda_{l}z_{l}}+
\frac{\sum_{l \in \mathbb Z_0\setminus\{1\}}\lambda_{l}z_{l}}{1+\lambda_1^2z_1^2+2\sum_{l \in \mathbb Z_0}\lambda_{l}z_{l}}=1.$$

Hence there exists a stationary distribution of the Markov chain corresponding to the measure $\mu_0$.

Due to the uniqueness of the stationary distribution from Theorem 2 in \cite{Sh}, (p.612) we get

\begin{cor} In the set of states $\mathbb Z$ of a Markov chain with the transition probabilities matrix $\mathbb P$, there is exactly one positive recurrent class of essential
	communicating states (for definitions, see Chapter VIII in \cite{Sh}).
\end{cor}

Summarizing, we have

\begin{thm}\label{t1} Let $k=2$. Then for the HC model with a countable set of states  (corresponding to the graph plotted in Fig.\ref{fig1}) the following statements are true:
\begin{itemize}
\item[1.] If the series $\sum_{j\in \mathbb Z_0}\lambda_j$ obtained from a sequence of parameters $\{\lambda_j\}_{j\in \mathbb Z_0}$ converges then there exists a unique translation-invariant Gibbs measure.

\item[2.] If the series $\sum_{j\in \mathbb Z_0}\lambda_j$  diverges there is no translation-invariant Gibbs measure.
\end{itemize}
\end{thm}

\section{Translation-invariant measures for the model with three loops}

We consider the graph $G$ with $a_{i0}=1$ for any $i\in \mathbb{Z}$, $a_{11}=1$, $a_{22}=1$  and $a_{im}=0$ otherwise (see Fig.\ref{fig2}).
\begin{figure}[h]
\begin{center}
   \includegraphics[width=13cm]{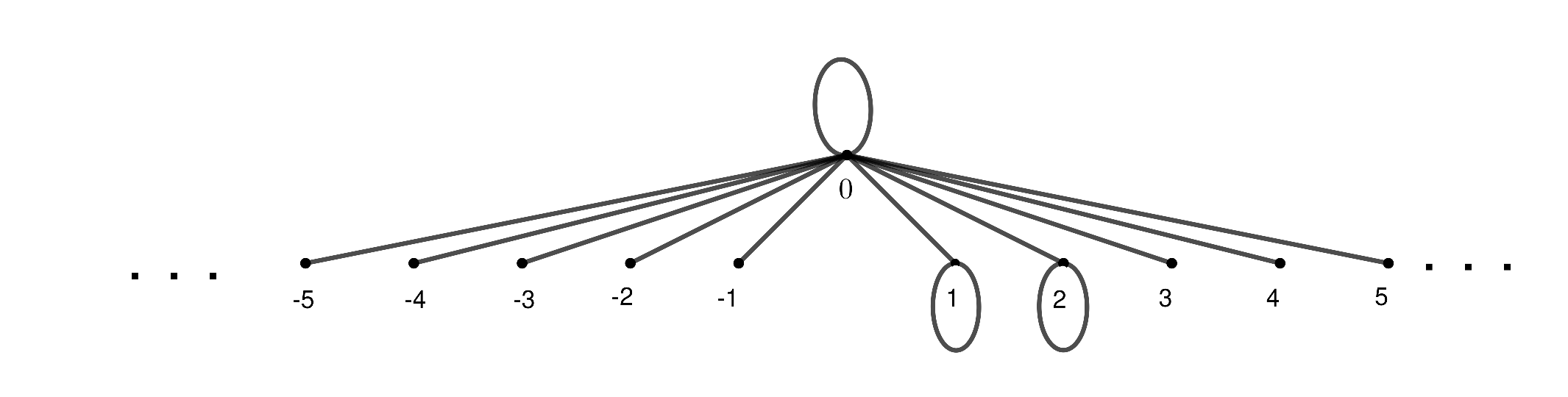}
\end{center}
     \caption{Countable graph $G$ with vertex set $\mathbb Z$ such that the loops are imposed at three vertices of the graph.}\label{fig2}
\end{figure}

For $x\in S(y)$, $y\in V$, introduce new variables as
$z_{i,x}=l_{xy}(i)$, then in case of $G$ (given in Fig.\ref{fig2}), from (\ref{ma}) we obtain
\begin{equation}\label{eq8}
\begin{cases}
z_{1,x}=\lambda_1\prod_{y \in S(x)} {1+z_{1,y}\over 1+\sum_{j\in \mathbb{Z}_0} z_{j,y}}, \\
z_{2,x}=\lambda_{2}\prod_{y \in S(x)} {1+z_{2,y}\over 1+\sum_{j\in \mathbb{Z}_0} z_{j,y}}, \\
z_{i,x}=\lambda_i\prod_{y \in S(x)} {1\over 1+\sum_{j\in \mathbb{Z}_0} z_{j,y}}, \ \ i\in \mathbb{Z}_0\setminus\{1,2\}.
\end{cases}
\end{equation}

\begin{rk}  As mentioned in Remark \ref{rk1}, in general, the location of the loops in the vertices (except vertex 0) of the graph does not matter.
\end{rk}

In this section, we consider translation invariant solutions of (\ref{eq8}), i.e., $z_x=z=(z_i)_{i\in \mathbb Z_0}$, with $z_i\in \mathbb R_{+}.$
In this case the equation (\ref{eq8}) has the following form
\begin{equation}\label{eq9}
\begin{cases}
z_{1}=\lambda_1\left({1+z_1\over 1+\sum_{j\in \mathbb{Z}_0} z_{j}}\right)^k, \\
z_{2}=\lambda_{2}\left({1+z_{2}\over 1+\sum_{j\in \mathbb{Z}_0} z_{j}}\right)^k, \\
z_i=\lambda_i\left({1\over 1+\sum_{j\in \mathbb Z_0} z_j}\right)^k, \ \ i\in \mathbb{Z}_0\setminus\{1,2\}.
\end{cases} \end{equation}
Here $\lambda_i>0, \ z_i>0$.

\begin{lemma}\label{L2} Let $k\geq1$. If there is a positive solution $\{z_j\}_{j\in \mathbb Z_0}$ of the system of equations (\ref{eq9}) for some sequence of parameters $\{\lambda_j\}_{j\in \mathbb Z_0}$ then series $\sum_{j\in \mathbb Z_0} z_j$ and $\sum_{j\in \mathbb Z_0}\lambda_j$ obtained  respectively from $\{z_i\}_{i\in \mathbb Z_0}$ and $\{\lambda_j\}_{j\in \mathbb Z_0}$ converge.
\end{lemma}

\begin{proof} The proof runs applying the very similar arguments as in Lemma \ref{L1}. In this case, we have
 $\sum_{i\in \mathbb Z_0} \lambda_i=A(1+A)^k-\lambda_1\left(\sum_{l=1}^kC_k^lz_1^l\right)-\lambda_{2}\left(\sum_{l=1}^kC_k^lz_{2}^l\right)$.
\end{proof}

  By Lemma \ref{L2} it follows that there is no positive solution of the system of equations (\ref{eq9}) for which the series $\sum_{j\in \mathbb Z_0} z_j$ and $\sum_{j\in \mathbb Z_0}\lambda_j$ diverge, i.e., these conditions are necessary for the existence of a solution (\ref{eq9}).

First, let us consider the first and second parts of the system of equations (\ref{eq9}):
\begin{equation}\label{eq10}
\begin{cases}
z_{1}=\lambda_1\left({1+z_1\over 1+\sum_{j\in \mathbb{Z}_0} z_{j}}\right)^k, \\
z_{2}=\lambda_{2}\left({1+z_{2}\over 1+\sum_{j\in \mathbb{Z}_0} z_{j}}\right)^k, \\
\end{cases} \end{equation}

In our further studies, it is very difficult to analyze the system of equations (\ref{eq10}) in the general case. Therefore, in (\ref{eq10}) we assume $\lambda=\lambda_{1}=\lambda_{2}$ and write the system of equations (\ref{eq10}) as follows:
\begin{equation}\label{eq11}
\begin{cases}
z_{1}=\lambda\left({1+z_1\over 1+\sum_{j\in \mathbb{Z}_0} z_{j}}\right)^k, \\
z_{2}=\lambda\left({1+z_{2}\over 1+\sum_{j\in \mathbb{Z}_0} z_{j}}\right)^k. \\
\end{cases} \end{equation}

Let $\Lambda_1=\frac{2\lambda^2+36\lambda-81}{\lambda}$ and $\Lambda_2=\frac{1}{1024}\left((18\lambda^2+64\lambda)\sqrt{9\lambda^2+32\lambda}+54\lambda^3+288\lambda^2+2304\lambda\right)$.

\begin{pro}\label{P2} Let $k=2$ and the series $\sum_{j\in \mathbb Z_0}\lambda_j$ obtained from the sequence of parameters $\{\lambda_j\}_{j\in \mathbb Z_0}$ converge and its sum $ \Lambda=\sum_{j\in \mathbb Z_0}\lambda_j$. Then the system of equations (\ref{eq9}):

1) For $\lambda\leq\frac{49}9$ and $\Lambda\geq\Lambda_1$ has a unique solution;

2) For $\lambda\leq\frac{49}9$ and $\Lambda<\Lambda_1$ has three solutions;

3) For $\lambda>\frac{49}9$ and $\Lambda\leq\Lambda_1$ has three solutions;

4) For $\lambda>\frac{49}9$ and $\Lambda_1<\Lambda<\Lambda_2$ has five solutions;

5) For $\lambda>\frac{49}9$ and $\Lambda=\Lambda_2$ has three solutions;

6) For $\lambda>\frac{49}9$ and $\Lambda>\Lambda_2$ has a unique solution

\end{pro}

\begin{proof} The proof will be given for the cases $z_1=z_2$ and $z_1\neq z_2$, separately.

\textbf{Case $z_1=z_2$}.
Let the series $\sum_{j\in \mathbb Z_0}\lambda_j$ converge and its sum be $\sum_{j\in \mathbb Z_0}\lambda_j=\Lambda$. It follows from Lemma 2 that for the existence of a solution $\{z_j\}_{j\in \mathbb Z_0}$ of the system of equations (\ref{eq9}), the series
$\sum_{j\in \mathbb Z_0} z_j$ converges .

Let $\sum_{j\in \mathbb Z_0} z_j=A$. Then due to (\ref{eq9}) we get
$$
\sum_{j\in \mathbb Z_0}z_j=\frac{\sum_{j\in \mathbb Z_0}\lambda_j+2\lambda\left(2z_1+z_1^2\right)}{(1+A)^2},
$$
i.e.,
\begin{equation}\label{eq13}
A(1+A)^2=\Lambda+2\lambda z_1\left(2+z_1\right), \ A>0.
\end{equation}

From the first and second parts of the system of equations (\ref{eq9}), the following can be rewritten:
\begin{equation}\label{eq14}
z_1(1+A)^2=\lambda\left(1+2z_1+z_1^2\right), \  \ \
\end{equation}

Solving (\ref{eq14}) with respect to $z_1$, we have:
$$z_1^{(1,2)}=\frac{(1+A)^2-2\lambda\pm(1+A)\sqrt{(1+A)^2-4\lambda}}{2\lambda}$$

It is easy to see that if $(1+A)^2>4\lambda$ then $D>0$, $z_1^{(1,2)}>0$ and $z_1^{(1)}= 1/z_1^{(2)}$.

On substituting $z_1 =z_2=z_1^{(1)}$ into (\ref{eq13}), we get:
$$
A(1+A)^2-\Lambda=\left((1+A)^2-2\lambda+(1+A)\sqrt{(1+A)^2-4\lambda}\right)\cdot\frac{(1+A)^2+2\lambda+(1+A)\sqrt{(1+A)^2-4\lambda}}{2\lambda}
$$
\begin{equation}\label{b1}\Rightarrow \ \ (1+A)^4+(1+A)^3\sqrt{(1+A)^2-4\lambda}-\lambda(1+A)^2(2+A)+\lambda(\Lambda-2\lambda)=0
\end{equation}

Consider the following function:
$$h(\lambda)=h(\lambda,x):=(1+x)^4+(1+x)^3\sqrt{(1+x)^2-4\lambda}-\lambda(1+x)^2(2+x)+\lambda(\Lambda-2\lambda), \ \ \ \lambda\in\left(0,\frac{(1+x)^2}{4}\right].$$
In this case, the equation \eqref{b1} can be written
\begin{equation}\label{b2}
h(\lambda,A)=0.
\end{equation}

We calculate the first and second derivatives of the function $h(\lambda)$ with respect to $\lambda$:
$$h'(\lambda)=-(x+1)^2(x+2)-\frac{2(x+1)^3}{\sqrt{(x+1)^2-4\lambda}}-4\lambda+\Lambda \quad  \ \ \ h''(\lambda)=-\frac{4(x+1)^3}{\sqrt{((x+1)^2-4\lambda)^3}}-4<0$$

It is easy to see that $h(\lambda)$ is a convex function, $h'(\lambda)$ exists in the interval $\lambda\in\left(0,\frac{(1+x)^2}{4}\right)$ and $h(0)=2(x+1)^4$. It follows that the function $h(\lambda)$ takes values smaller than $2(x+1)^4$ at most once. Therefore, the number of the solution of the equation $h(\lambda)=0$ is at most one.

Let us put $z_1=z_2=z_1^{(2)}$ in (\ref{eq13}). In that case:
$$
A(1+A)^2-\Lambda=\left((1+A)^2-2\lambda-(1+A)\sqrt{(1+A)^2-4\lambda}\right)\cdot\frac{(1+A)^2+2\lambda-(1+A)\sqrt{(1+A)^2-4\lambda}}{2\lambda}
$$
\begin{equation}\label{b3}
\Rightarrow \ \ (1+A)^4-(1+A)^3\sqrt{(1+A)^2-4\lambda}-\lambda(1+A)^2(2+A)+\lambda(\Lambda-2\lambda)=0
\end{equation}
Consider the following function:
$$\delta(\lambda)=(1+x)^4-(1+x)^3\sqrt{(1+x)^2-4\lambda}-\lambda(1+x)^2(2+x)+\lambda(\Lambda-2\lambda), \ \ \ \ \lambda\in\left(0,\frac{(1+x)^2}{4}\right].$$

In this case, the equation \eqref{b3} can be written
\begin{equation}\label{b4}
\delta(\lambda,A)=0.
\end{equation}

We calculate the first and second derivatives of the function $\delta(\lambda)$ with respect to $\lambda$:
$$\delta'(\lambda)=-(x+1)^2(x+2)+\frac{2(x+1)^3}{\sqrt{(x+1)^2-4\lambda}}-4\lambda+\Lambda \quad  \ \ \ \delta''(\lambda)=\frac{4(x+1)^3}{\sqrt{((x+1)^2-4\lambda)^3}}-4<0.$$

It is obvious that $\delta(\lambda)$ is a concave function, $\delta'(\lambda)$  exists in $\lambda\in\left(0,\frac{(x+1)^2}{4}\right)$ and $\delta(0)=0$. Therefore, the equation $\delta(\lambda)=0$ has at most one positive solution.

On the other hand, the function
$$\psi(\lambda):=h(\lambda)-\delta(\lambda)=2(x+1)^3\sqrt{(x+1)^2-4\lambda}$$
is decreasing with respect to $\lambda$ and $\psi\Big(\frac{(x+1)^2}4\Big )=$0, i.e., $h\Big(\frac{(x+1)^2}{4}\Big)=\delta\Big(\frac{(x+1)^2}{4}\Big)$. As a result, the graph of the functions $h(\lambda)$ and $\delta(\lambda)$ together with the interval $[0;2(x+1)^4]$ of the $OY$ axis forms a closed line (see Fig.\ref{loop2} ). Then only one of the equations $h(\lambda)=0$ and $\delta(\lambda)=0$ would have a positive solution. If $h\Big(\frac{(x+1)^2}{4}\Big)>0$ then the equation $\delta(\lambda)=0$ has a positive solution, if $h\Big(\frac{(x+1)^2}{4}\Big)<0$ then the equation $h(\lambda)=0$ has a positive solution. Hence, there is a unique $x_0$ corresponding to each $\lambda$. Accordingly, there is a unique $A_0$ corresponding to each $\lambda$.

\begin{figure}[h]
\begin{center}
   \includegraphics[width=5cm]{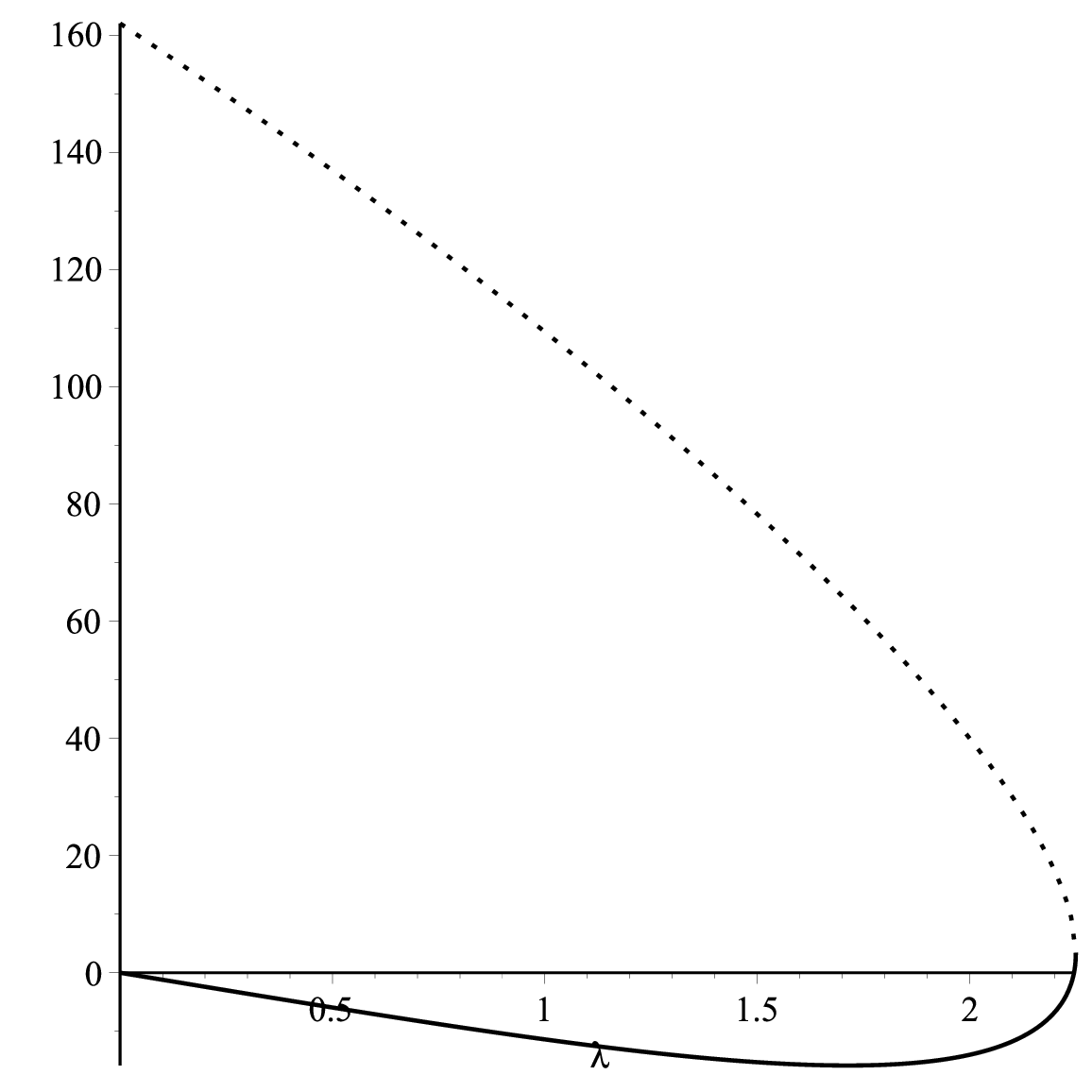}
\end{center}
     \caption{Graph functions $h(\lambda)$ (dotted line) and $\delta(\lambda)$ (solid line) for $x=2$ and $\Lambda=10$.} \label{loop2}
\end{figure}

Then the unique positive solution of the system of equations (\ref{eq9}) is determined as follows:

If the equation $h(\lambda)=0$ has a solution then we take $z_1=z_2=z_1^{(1)}$ as a solution, if the equation $\delta(\lambda)=0$ has a solution then we take $z_1=z_2=z_1^{(2)}$ as a solution. The remaining coordinates are found as follows:
$$z_i={\lambda_i \over\left( 1+A_0\right)^2}, \ \ i\in \mathbb{Z}_0\setminus\{1,2\}.$$

\textbf{Case $z_1\neq z_2$}. Let the series $\sum_{j\in \mathbb Z_0}\lambda_j$ converge and its sum be $\sum_{j\in \mathbb Z_0}\lambda_j=\Lambda$. It follows from Lemma 2 that for the existence of a solution $\{z_j\}_{j\in \mathbb Z_0}$ of the system of equations (\ref{eq9}), the series $\sum_{j\in \mathbb Z_0} z_j$ converges .

Let $\sum_{j\in \mathbb Z_0} z_j=A$. Then due to (\ref{eq9}) we get
$$
\sum_{j\in \mathbb Z_0}z_j=\frac{\sum_{j\in \mathbb Z_0}\lambda_j+\lambda\left(2z_1+z_1^2+2z_2+z_2^2\right)}{(1+A)^2},
$$
i.e.,
\begin{equation}\label{e18}
A(1+A)^2=\Lambda+\lambda\left(2z_1+z_1^2+2z_2+z_2^2\right), \ A>0.
\end{equation}

The first and second parts of the system of equations (\ref{eq9}) can be rewritten:
\begin{equation}\label{e19}
z_1(1+A)^2=\lambda\left(1+2z_1+z_1^2\right), \  \ \
z_2(1+A)^2=\lambda\left(1+2z_2+z_2^2\right).
\end{equation}

Solving (\ref{e19}) with respect to $z_1$ (resp. $z_2$):
$$z^{(1,2)}=\frac{(1+A)^2-2\lambda\pm(1+A)\sqrt{(1+A)^2-4\lambda}}{2\lambda}$$

It is easy to see that if $(1+A)^2>4\lambda$ then $D>0$, $z_t^{(1,2)}>0$ and $z^{(1)}= 1/z^{(2)}$.

On substituting $z_1=z^{(1)}$ and $z_2=z^{(2)}$ into (\ref{e18}), we obtain
$$
A(1+A)^2=\Lambda+\frac1{\lambda}\Big((1+A)^4-2\lambda(1+A)^2-2\lambda^2\Big) \ \ \ \Rightarrow
$$
\begin{equation}\label{e20}
(1+A)^4-\lambda(A+2)(1+A)^2+\lambda\Lambda-2\lambda^2=0
\end{equation}

Consider the following function:
$$q(x)=q(x,\lambda):=(1+x)^4-\lambda(x+2)(1+x)^2+\lambda\Lambda-2\lambda^2. $$

We calculate the first derivative of the function $q(x)$ with respect to $x$:
$$q'(x)=4x^3+(12-3\lambda)x^2+(12-8\lambda)x+4-5\lambda.$$

The solutions of $q'(x)=0$ are
$$x_1=-1, \ \ x_2=\frac{3\lambda-8-\sqrt{9\lambda^2+32\lambda}}{8}, \ \ \ x_3=\frac{3\lambda-8+\sqrt{9\lambda^2+32\lambda}}{8}$$
Since $(A+1)^2>4\lambda$, it suffices to check the function $q(x)$ in the interval $x>2\sqrt{\lambda}-1$.

If $\lambda>0$ then $x_2<0$, and if $\lambda>\frac{49}{9}$, then $x_3>2\sqrt{\lambda}-1$. Therefore, if $\lambda>\frac{49}{9}$ the function $q(x)$ is increasing for $x>x_3$, and decreasing for $2\sqrt{\lambda}-1<x<x_3$.

$(i)$ Let  $\lambda\leq\frac{49}9$. Then the function $q(x)$ is increasing for $x>2\sqrt{\lambda}-1$. It is easy to verify that
if $\Lambda<\Lambda_1$ then the equation $q(x)=0$ has a unique positive solution for $x>2\sqrt{\lambda}-1$, i.e.,
$$q\left(2\sqrt{\lambda}-1\right)=\lambda\Lambda+10\lambda^2-8\lambda^{5/2}<0, \quad \Rightarrow \quad \Lambda<\Lambda_1=8\lambda^{3/2}-10\lambda.$$

If $\Lambda\geq\Lambda_1$ then the equation $q(x)=0$ does not have a positive solution in the interval $x>2\sqrt{\lambda}-1$, i.e.,
$$q\left(2\sqrt{\lambda}-1\right)=\lambda\Lambda+10\lambda^2-8\lambda^{5/2}\geq0, \quad \Rightarrow \quad \Lambda\geq\Lambda_1=8\lambda^{3/2}-10\lambda.$$

$(ii)$ Let $\lambda>\frac{49}9$. Then the function $q(x)$ is increasing for $x>x_3$, and decreasing for $2\sqrt{\lambda}-1<x<x_3$.

At the point $x=x_3$, it reaches its local minimum in the interval $x>2\sqrt{\lambda}-1$, and
$$q(x_3)=-\frac{\lambda}{512}\left((9\lambda^2+32\lambda)\sqrt{9\lambda^2+32\lambda}+27\lambda^3+144\lambda^2-512\Lambda+1152\lambda\right)$$
On the other hand,
 $$q\left(2\sqrt{\lambda}-1\right)=\lambda\Lambda+10\lambda^2-8\lambda^{5/2}, \ \ q(\infty)=\infty.$$

Assume that $\Lambda_1=8\lambda^{3/2}-10\lambda$ and $\Lambda_2=\frac{1}{512}\left((9\lambda^2+32\lambda)\sqrt{9\lambda^2+32\lambda}+27\lambda^3+144\lambda^2+1152\lambda\right)$ bo'lsin.
Obviously, if $\lambda>\frac{49}9$, then $\Lambda_1<\Lambda_2$.
Then we have the following assertions:

$(i)$ If $q\left(2\sqrt{\lambda}-1\right)\leq0$ and $q(x_3)<0$, i.e. $\Lambda\leq\Lambda_1$, then the equation $q(x)=0$ has a unique positive solution for $x>2\sqrt{\lambda}-1$;

$(ii)$ If $q\left(2\sqrt{\lambda}-1\right)>0$ and $q(x_3)<0$, i.e. $\Lambda_1<\Lambda<\Lambda_2$, then the equation $q(x)=0$ has 2 positive solutions for $x>2\sqrt{\lambda}-1$;

$(iii)$ If $q\left(2\sqrt{\lambda}-1\right)>0$ and $q(x_3)=0$, i.e. $\Lambda=\Lambda_2$, then $q( x)=0$ has a unique positive solution for $x>2\sqrt{\lambda}-1$.

$(iv)$ If $q\left(2\sqrt{\lambda}-1\right)>0$ and $q(x_3)>0$, i.e. $\Lambda>\Lambda_2$, then the equation $q(x)=0$ does not have a positive solution for $x>2\sqrt{\lambda}-1$.

Suppose that the solutions of the equation $q(x)=0$ are $A_1$ and $A_2$. Then the solutions of the system of equations (\ref{eq9}) are:
 $$z_1=\frac{(1+A_1)^2-2\lambda+(1+A_1)\sqrt{(1+A_1)^2-4\lambda}}{2\lambda}, \ \ z_2=\frac{1}{z_1}, \ \ z_i={\lambda_i \over\left( 1+A_1\right)^2}, \ \ i\in \mathbb{Z}_0\setminus\{1,2\}.$$
  $$z_1=\frac{(1+A_2)^2-2\lambda+(1+A_2)\sqrt{(1+A_2)^2-4\lambda}}{2\lambda}, \ \ z_2=\frac{1}{z_1},  \ \ z_i={\lambda_i \over\left( 1+A_2\right)^2}, \ \ i\in \mathbb{Z}_0\setminus\{1,2\}.$$

Applying above process to $z_1=z^{(2)}$ and $z_{2}=z^{(1)}$, respectively, we obtain:

suppose that the solutions of the equation $q(x)=0$ are $A_1$ and $A_2$. Then the solutions of the system of equations (\ref{eq9}) are:
 $$z_1=\frac{(1+A_1)^2-2\lambda-(1+A_1)\sqrt{(1+A_1)^2-4\lambda}}{2\lambda}, \ \  z_2=\frac{1}{z_1}, \ \ z_i={\lambda_i \over\left( 1+A_1\right)^2}, \ \ i\in \mathbb{Z}_0\setminus\{1,2\}.$$
  $$z_1=\frac{(1+A_2)^2-2\lambda-(1+A_2)\sqrt{(1+A_2)^2-4\lambda}}{2\lambda}, \ \ z_2=\frac{1}{z_1},  \ \ z_i={\lambda_i \over\left( 1+A_2\right)^2}, \ \ i\in \mathbb{Z}_0\setminus\{1,2\}.$$
\end{proof}

\begin{rk}  We note that the solution  $\{z_j\}_{j\in \mathbb Z_0}$ in Proposition \ref{P2} is normalisable because the convergence of series $\sum_{j\in \mathbb Z_0} z_j^{\frac{k+1}{k}}$ follows from the convergence of $\sum_{j\in \mathbb Z_0} z_j$. Then by Remark \ref{r2}  the Gibbs measure (denoted by $\mu_i$, $i=1...5$) corresponding to this solution exists.
\end{rk}

\section{Markov chain corresponding to the obtained TIGMs in Section 5}

For the Gibbs measure $\mu_i$, $i=1...5$ we define a matrix $\mathbb P_1$ of transition probabilities. We construct the matrix $\mathbb P_1$ using the very similar arguments as in Section 4.

For the considered model, $a_{11}=1$, $a_{22}$=1, $a_{i0}=1$, $a_{0j}=1$ and $a_{ij}=0$ for $i\in \mathbb Z_0$ and $j \in \mathbb Z_0$. Hence
$$p_{00}=\frac{1}{1+\sum_{l \in \mathbb Z_0}{\lambda'_lz'_l}}, \ \ p_{11}=\frac{\lambda'_1z'_1}{1+\lambda'_1z'_1}, \ \ p_{22}=\frac{\lambda'_2z'_2}{1+\lambda'_2z'_2}, \ \ p_{01}=\frac{\lambda'_1 z'_1}{1+\sum_{l \in \mathbb Z_0}{\lambda'_lz'_l}}, $$
$$ p_{0i}=\frac{\lambda_i z'_i}{1+\sum_{l \in \mathbb Z_0}{\lambda_lz'_l}},\ \  p_{10}=\frac{1}{1+\lambda'_1z'_1}, \ \ p_{20}=\frac{1}{1+\lambda'_2z'_2}, \ \ p_{i0}=1.$$

Therefore $\mathbb P_1$ has the following form:
$$\mathbb P_1=\begin{pmatrix}
\cdots & \vdots & \vdots & \vdots & \vdots & \vdots & \cdots \\
 \cdots & 0 & 0 & 1  &  0 &  0 & \cdots \\
 \cdots & 0 & 0 & 1  & 0 &  0 & \cdots \\
 \cdots & \frac{\lambda_{-2} z_{-2}}{1+\sum_{l \in \mathbb Z_0}{\lambda_lz_l}} & \frac{\lambda_{-1} z_{-1}}{1+\sum_{l \in \mathbb Z_0}{\lambda_lz_l}} &  \frac{1}{1+\sum_{l \in \mathbb Z_0}{\lambda_lz_l}}  &  \frac{\lambda_1 z_1}{1+\sum_{l \in \mathbb Z_0}{\lambda_lz_l}}&  \frac{\lambda_2 z_2}{1+\sum_{l \in \mathbb Z_0}{\lambda_lz_l}} & \cdots \\
 \cdots & 0 & 0 & \frac{1}{1+\lambda_1z_1}  &   \frac{\lambda_1z_1}{1+\lambda_1z_1} &  0 & \cdots \\
 \cdots & 0 & 0 & \frac{1}{1+\lambda_2z_2}  &  0 &  \frac{\lambda_2z_2}{1+\lambda_2z_2} & \cdots \\
  \cdots & 0 & 0 & 1  &  0 &  0 & \cdots \\
 \cdots & \vdots & \vdots & \vdots & \vdots & \vdots & \cdots
 \end{pmatrix}$$

We consider the vector $X=(\cdots,x_{-2},x_{-1},x_0,x_1,x_2,\cdots), \ \sum_{j\in \mathbb Z}{x_j}=1.$ We solve the equation $X\cdot\mathbb P_1=X$. We have
$$X\cdot \mathbb P_1=\left(\cdots,\frac{x_0\lambda_{-2}z_{-2}}{1+\sum_{l \in \mathbb Z_0}{\lambda_lz_l}},\frac{x_0\lambda_{-1}z_{-1}}{1+\sum_{l \in \mathbb Z_0}{\lambda_lz_l}},p_{0},p_{1},p_{2}, \frac{x_0\lambda_{3}z_{3}}{1+\sum_{l \in \mathbb Z_0}{\lambda_lz_l}},\cdots\right),$$
where
$$p_{0}=1-\frac{x_0\sum_{l \in \mathbb Z_0}{\lambda_lz_l}}{1+\sum_{l \in \mathbb Z_0}{\lambda_lz_l}}-\frac{x_1\lambda_1z_1}{1+\lambda_1z_1}-\frac{x_2\lambda_2z_2}{1+\lambda_2z_2}$$
$$p_{1}=\frac{x_0\lambda_1z_{1}}{1+\sum_{l \in \mathbb Z_0}{\lambda_lz_l}}+\frac{x_1\lambda_1z_1}{1+\lambda_1z_1}, \ \ \
p_{2}=\frac{x_0\lambda_1z_{1}}{1+\sum_{l \in \mathbb Z_0}{\lambda_lz_l}}+\frac{x_2\lambda_2z_2}{1+\lambda_2z_2}$$

From the equality $X\cdot\mathbb P_1=X$ we get
$$x_0=1-\frac{x_0\sum_{l \in \mathbb Z_0}{\lambda_lz_l}}{1+\sum_{l \in \mathbb Z_0}{\lambda_lz_l}}-\frac{x_1\lambda_1z_1}{1+\lambda_1z_1}-\frac{x_2\lambda_2z_2}{1+\lambda_2z_2}, \ \  x_1=\frac{x_0\lambda_1z_{1}}{1+\sum_{l \in \mathbb Z_0}{\lambda_lz_l}}+\frac{x_1\lambda_1z_1}{1+\lambda_1z_1},$$
\begin{equation}\label{e21}
x_2=\frac{x_0\lambda_1z_{1}}{1+\sum_{l \in \mathbb Z_0}{\lambda_lz_l}}+\frac{x_2\lambda_2z_2}{1+\lambda_2z_2}, \ \ x_j=\frac{x_0\lambda_{j}z_{j}}{1+\sum_{l \in \mathbb Z_0}{\lambda_lz_l}},
\end{equation}
where $j\in \mathbb Z_0\setminus\{1,2\}.$

From the equality of (\ref{e21}) we  find $x_0$, $x_1$ and $x_2$:
$$x_0=\frac{1+\sum_{l \in \mathbb Z_0}\lambda_{l}z_{l}}{1+\lambda_1^2z_1^2+\lambda_2^2z_2^2+2\sum_{l \in \mathbb Z_0}\lambda_{l}z_{l}},$$
$$x_1=\frac{\lambda_1^2z_1^2+\lambda_1z_1}{1+\lambda_1^2z_1^2+\lambda_1^2z_1^2+2\sum_{l \in \mathbb Z_0}\lambda_{l}z_{l}}, \ x_2=\frac{\lambda_2^2z_2^2+\lambda_2z_2}{1+\lambda_1^2z_1^2+\lambda_1^2z_1^2+2\sum_{l \in \mathbb Z_0}\lambda_{l}z_{l}}$$

Using the expression for $x_0$ from the fourth equality of (\ref{e21}) we find $x_j$ $ j\in \mathbb Z_0\setminus\{1,2\}$:
$$x_j=\frac{\lambda_{j}z_{j}}{1+\lambda_1^2z_1^2+\lambda_2^2z_2^2+2\sum_{l \in \mathbb Z_0}\lambda_{l}z_{l}}.$$

It is easy to see that the obtained vector is stochastic:
$$\sum_{j\in \mathbb Z}{x_j}=\frac{1+\sum_{l \in \mathbb Z_0}\lambda_{l}z_{l}}{1+\lambda_1^2z_1^2+\lambda_2^2z_2^2+2\sum_{l \in \mathbb Z_0}\lambda_{l}z_{l}}+\frac{\lambda_1^2z_1^2+\lambda_1z_1}{1+\lambda_1^2z_1^2+\lambda_2^2z_2^2+2\sum_{l \in \mathbb Z_0}\lambda_{l}z_{l}}+$$
$$
+\frac{\lambda_2^2z_2^2+\lambda_2z_2}{1+\lambda_1^2z_1^2+\lambda_2^2z_2^2+2\sum_{l \in \mathbb Z_0}\lambda_{l}z_{l}}+\frac{\sum_{l \in \mathbb Z_0\setminus\{1,2\}}\lambda_{l}z_{l}}{1+\lambda_1^2z_1^2+\lambda_2^2z_2^2+2\sum_{l \in \mathbb Z_0}\lambda_{l}z_{l}}=1.$$

Hence there exists a stationary distribution of the Markov chain corresponding to the measure $\mu_i$, $i=1...5$.

Due to the uniqueness of the stationary distribution from Theorem 2 in \cite{Sh}, (p.612) we get

\begin{cor} In the set of states $\mathbb Z$ of a Markov chain with the transition probabilities matrix $\mathbb P_1$, there is exactly one positive recurrent class of essential
	communicating states (for definitions, see Chapter VIII in \cite{Sh}).
\end{cor}

Summarizing, we have

\begin{thm}\label{t2} Let $k=2$, $\Lambda_2=\frac{1}{1024}\left((18\lambda^2+64\lambda)\sqrt{9\lambda^2+32\lambda }+54\lambda^3+288\lambda^2+2304\lambda\right)$ and $\Lambda_1=8\lambda^{3/2}-10\lambda$. Then for the NN model with a countable number of states (corresponding to the graph from Fig. \ref{fig2}), the following statements are true:

1.If the series $\sum_{j\in \mathbb Z_0}\lambda_j$ obtained from the terms of the sequence of parameters $\{\lambda_j\}_{j\in \mathbb Z_0}$ converges and its sum $\sum_{j\in \mathbb Z_0}\lambda_j=\Lambda$, then

$(i)$ For $\lambda\leq\frac{49}9$ and $\Lambda<\Lambda_1$, there are exactly three TIGMs;

$(ii)$ For $\lambda\leq\frac{49}9$ and $\Lambda\geq\Lambda_1$, there is exactly one TIGM;

$(iii)$ For $\lambda>\frac{49}9$ and $\Lambda\leq\Lambda_1$, there are exactly three TIGMs;

$(iv)$ For $\lambda>\frac{49}9$ and $\Lambda_1<\Lambda<\Lambda_2$, there are exactly five TIGMs;

$(v)$ For $\lambda>\frac{49}9$ and $\Lambda=\Lambda_2$, there are exactly three TIGMs;

$(vi)$ For $\lambda>\frac{49}9$ and $\Lambda>\Lambda_2$ there is exactly one TIGM

2. If the series $\sum_{j\in \mathbb Z_0}\lambda_j$ diverges, then there is no TIGM.
\end{thm}

\section*{ Acknowledgements}

The work supported by the fundamental project (number: F-FA-2021-425)  of The Ministry of Innovative Development of the Republic of Uzbekistan.

\section*{Statements and Declarations}

{\bf	Conflict of interest statement:}
On behalf of all authors, the corresponding author states that there is no conflict of interest.

\section*{Data availability statements}
The datasets generated during and/or analysed during the current study are available
from the corresponding author on reasonable request.

\end{document}